%% file: kmed.tex
\documentclass[a4paper,UKenglish,cleveref, autoref, thm-restate]{lipics-v2021}
\usepackage{amsmath,amsfonts,amssymb,amsthm}
\usepackage{mathbbol}
\usepackage{amsthm}
\usepackage{graphicx,color}
\usepackage{boxedminipage}
\usepackage[ruled, linesnumbered, boxed,lined]{algorithm2e}
\usepackage{algorithmicx}
\usepackage{framed}
\usepackage{thmtools}
\usepackage{thm-restate}
\usepackage{xspace}
\usepackage{todonotes}
\usepackage{mathtools}

\usepackage[
                 ]{hyperref}

\usepackage{xifthen}
\usepackage{tabularx}
\usepackage{arydshln}
\usetikzlibrary{calc}

\newcommand{\Oh}{\mathcal{O}}

\DeclareMathOperator{\operatorClassP}{{\sf P}}
\newcommand{\classP}{\ensuremath{\operatorClassP}}
\DeclareMathOperator{\operatorClassNP}{{\sf NP}}
\newcommand{\classNP}{\ensuremath{\operatorClassNP}}

\newcommand{\zfill}{Z_{\mathsf{fill}}}
\newcommand{\uiso}{U_{\mathsf{iso}}}

\crefname{invar}{invariant}{invariants}
\crefname{ineq}{inequality}{inequalities}
\crefname{constr}{constraint}{constraints}
\crefname{tbl}{table}{tables}
\crefname{lem}{lemma}{lemmata}
\crefname{lemma}{lemma}{lemmata}
\crefname{theorem}{theorem}{theorems}
\crefname{cond}{condition}{conditions}

\newcommand{\pname}{\textsc}
\newcommand{\ProblemFormat}[1]{\pname{#1}}
\newcommand{\ProblemIndex}[1]{\index{problem!\ProblemFormat{#1}}}
\newcommand{\ProblemName}[1]{\ProblemFormat{#1}\ProblemIndex{#1}{}\xspace}

\newcommand{\probMedian}{\ProblemName{$k$-median}} 
\newcommand{\probMeans}{\ProblemName{$k$-means}} 
\newcommand{\probCenter}{\ProblemName{$k$-center}} 
\newcommand{\probMedianFacility}{\ProblemName{$k$-median Facility Location}}  
\newcommand{\probDominating}{\ProblemName{Dominating set}}     
\newcommand{\probPM}{\ProblemName{Minimum Weight Perfect Matching}}
\newcommand{\probSetC}{\ProblemName{Set Cover}} 
\newcommand{\probDeltaSetC}{\ProblemName{$\Delta$-Set Cover}} 
\newcommand{\probSeCoCo}{\ProblemName{Set Cover Conjecture}} 
\newcommand{\probMWSetC}{\ProblemName{Minimum weighted $k$-set cover}}

\newcommand{\OPT}{\mathsf{OPT}}
\newcommand{\cost}{\mathsf{cost}}
\newcommand{\kmed}{\mathsf{kmed}}
\newcommand{\mm}{\mathsf{mm}}
\newcommand{\red}[1]{{\color{red}#1}}
\newcommand{\poly}{\mathsf{poly}}



\newlength{\RoundedBoxWidth}
\newsavebox{\GrayRoundedBox}
\newenvironment{GrayBox}[1]%
   {\setlength{\RoundedBoxWidth}{.93\textwidth}
    \def\boxheading{#1}
    \begin{lrbox}{\GrayRoundedBox}
       \begin{minipage}{\RoundedBoxWidth}}%
   {   \end{minipage}
    \end{lrbox}
    \begin{center}
    \begin{tikzpicture}%
       \node(Text)[draw=black!20,fill=white,rounded corners,%
             inner sep=2ex,text width=\RoundedBoxWidth]%
             {\usebox{\GrayRoundedBox}};
        \coordinate(x) at (current bounding box.north west);
        \node [draw=white,rectangle,inner sep=3pt,anchor=north west,fill=white] 
        at ($(x)+(6pt,.75em)$) {\boxheading};
    \end{tikzpicture}
    \end{center}}     

\newenvironment{defproblemx}[2][]{\noindent\ignorespaces%
                                \FrameSep=6pt%
                                \parindent=0pt%
                \vspace*{-1.5em}
                \ifthenelse{\isempty{#1}}{%
                  \begin{GrayBox}{\textsc{#2}}%
                }{%
                  \begin{GrayBox}{\textsc{#2} parameterized by~{#1}}%
                }
                \begin{tabular*}{\textwidth}{@{\hspace{.1em}} >{\itshape} p{1.8cm} p{0.8\textwidth} @{}}%
            }{
                \end{tabular*}%
                \end{GrayBox}%
                \ignorespacesafterend
            }

\newcommand{\defproblema}[3]{
  \begin{defproblemx}{#1}
    Input:  & #2 \\
    Task: & #3
  \end{defproblemx}
}%

\title{Exact Exponential Algorithms for Clustering  Problems} 

\titlerunning{} 

\author{Fedor V. Fomin}{Department of Informatics, University of Bergen, Norway}{Fedor.Fomin@uib.no}{https://orcid.org/0000-0003-1955-4612}{}
\author{Petr A. Golovach}{Department of Informatics, University of Bergen, Norway}{Petr.Golovach@uib.no }{https://orcid.org/0000-0002-2619-2990}{}

\author{Tanmay Inamdar \footnote[1]{Part of this work was done when the two authors were visiting IMSc, Chennai.}}{Department of Informatics, University of Bergen, Norway}{Tanmay.Inamdar@uib.no}{}{}

\author{Nidhi Purohit \footnotemark[1]{}}{Department of Informatics, University of Bergen, Norway}{Nidhi.Purohit@uib.no}{}{}

\author{Saket Saurabh}{The Institute of Mathematical Sciences, HBNI, Chennai, India and Department of Informatics, University of Bergen, Norway}{saket@imsc.res.in}{}{}

\authorrunning{ }

\Copyright{Fedor V. Fomin, Petr A. Golovach, Tanmay Inamdar, Nidhi Purohit, Saket Saurabh}


\funding{The research leading to these results has received funding from the Research Council of Norway via the project  BWCA
(grant no. 314528) and the European Research Council (ERC) via grant LOPPRE, reference 819416.}


\category{} 

\relatedversion{} 

\supplement{}

\acknowledgements{}

\nolinenumbers 

\hideLIPIcs  

\EventEditors{John Q. Open and Joan R. Access}
\EventNoEds{2}
\EventLongTitle{42nd Conference on Very Important Topics (CVIT 2016)}
\EventShortTitle{CVIT 2016}
\EventAcronym{CVIT}
\EventYear{2016}
\EventDate{December 24--27, 2016}
\EventLocation{Little Whinging, United Kingdom}
\EventLogo{}
\SeriesVolume{42}
\ArticleNo{23}

\begin{document}

\maketitle

\begin{abstract}
In this paper we initiate a systematic study of exact algorithms for some of the well known clustering problems, namely \probMedian and \probMeans. In \probMedian, the input consists of a set $X$ of $n$ points belonging to a metric space, and the task is to select a subset $C \subseteq X$ of $k$ points as \emph{centers}, such that the sum of the distances of every point to its nearest center is minimized. In \probMeans, the objective is to minimize the sum of \emph{squares} of the distances instead. It is easy to design an algorithm running in time $\max_{k\leq n} {n \choose k} n^{\Oh(1)} = \Oh^*(2^n)$ (here, $\Oh^*(\cdot)$ notation hides polynomial factors in $n$). In this paper we design first non-trivial exact algorithms for these problems. In particular, we obtain an $\Oh^*((1.89)^n)$ time \emph{exact} algorithm for \probMedian that works for any value of $k$. Our algorithm is quite general in that it does not use any properties of the underlying (metric) space -- it does not even require the distances to satisfy the triangle inequality. In particular, the same algorithm also works for $k$-\textsc{Means}. We complement this result by showing that the running time of our algorithm is asymptotically optimal, up to the base of the exponent. That is, unless the Exponential Time Hypothesis fails, there is no algorithm for these problems running in time $2^{o(n)} \cdot n^{\Oh(1)}$.
 
Finally, we consider the ``facility location'' or ``supplier'' versions of these clustering problems, where, in addition to the set $X$ we are additionally given a set of $m$ candidate centers (or facilities) $F$, and objective is to find a subset of $k$ centers from $F$. The goal is still to minimize the $k$-\textsc{Median}/$k$-\textsc{Means}/$k$-\textsc{Center} objective. For these versions we give a $\Oh(2^n (mn)^{\Oh(1)})$ time algorithms using subset convolution. We complement this result by showing that, under the Set Cover Conjecture, the ``supplier'' versions of these problems do not admit an exact algorithm running in time $2^{(1-\epsilon) n} (mn)^{\Oh(1)}$.

\end{abstract}

\newpage

\input{introduction}
\input{preliminaries}
\input{matching}
\input{eth-hard}
\input{secoco-hard}

\input{fl-exact}



\bibliography{kmed}

\end{document}

%% file: introduction.tex
\section{Introduction}\label{sec:intro}

Clustering is a fundamental area in the domain of optimization problems with numerous applications. In this paper, we focus on  some of the most fundamental problems in the clustering literature, namely \probMedian, \probMeans, and \probCenter. We formally define the optimization version \probMedian.

\defproblema{\probMedian}%
{Given a metric space $(X,d)$, where $X=\{x_1,\ldots,x_n\}$ is a collection of $n$ points, with distance function $d$ on $X$ and a positive integer $k$.}%
{Find a pair $(C,P)$, where $C=\{c_1,\ldots,c_k\} \subseteq X$ is a set of \emph{centers} and $P$ is a partition of $X$ into $k$ subsets $\{X_1,\ldots,X_k\}$ (clusters). Here, $X_i$ is the cluster corresponding to the center $c_i \in C$. The goal is to minimize the following cost, over all pairs $(C, P)$.
	$$\cost(C,P)=\sum_{i=1}^{k}\sum_{x \in X_i}^{}d(c_i,x)$$\vspace{-0.5cm}}
\probMeans is a variant of \probMedian, where the only difference is that we want to minimize the sum of squares of the distances, i.e., $\displaystyle \sum_{i = 1}^k \sum_{x \in X_i} (d(c_i, x))^2$. In \probCenter, the objective is to minimize the maximum distance of a point and its nearest center, i.e., $\displaystyle \max_{i = 1}^k \max_{x \in X_i} d(c_i, x)$. 

The special cases of \probMedian have a long history, and they are known in the literature as Fermat-Weber problem \cite{wiki:median,wiki:weber}. A recent formulation of \probMeans can be traced back to Steinhaus \cite{steinhaus1956division} and MacQueen \cite{macqueen1967some}. Lloyd proposed a heuristic algorithm  \cite{lloyd1982least} for \probMeans that is extremely simple to implement for euclidean spaces, and it remains popular even today. \probCenter was proved to be \classNP-complete by Hsu and Nemhauser \cite{hsu1979easy}. All three problems have been studied from the perspective of approximation algorithms for last several decades. These three problems---as well as several of their generalizations---are known to admit constant factor approximations in polynomial time. More recently, these problems have also been studied from the perspective of Fixed-Parameter Tractable (FPT) algorithms, where one allows the running times of the form $f(k) \cdot n^{\Oh(1)}$ for some computable function $f$. \probMedian and \probMeans are known to admit improved approximation guarantees using FPT algorithms \cite{cohen2019tight}, and these approximation guarantees are tight up to certain complexity-theoretic assumptions.

A result that initiated this study is an \emph{exact} algorithm for \probCenter \footnote{We note that the result of \cite{agarwal2002exact} holds for a slightly different variant, where the centers can be placed anywhere in $\mathbb{R}^d$. This formulation is more natural and standard in euclidean spaces.} by Agarwal and Procopiuc \cite{agarwal2002exact}, who give an $n^{\Oh(k^{1-\frac{1}{d}})}$ time algorithm in $\mathbb{R}^d$. In particular, in two dimensional space, their algorithm runs in $2^{\Oh(\sqrt{n} \log n)}$ time for any value of $k$, i.e., in \emph{sub-exponential} time. This led us towards a natural question, namely, studying the complexity of \probMedian, \probMeans, \probCenter in general metrics. 

Note that it is easy to design an exact algorithm that runs in time $\binom{n}{k} \cdot n^{\Oh(1)}$ -- it simply enumerates all sets of centers of size $k$, and the corresponding partition of $X$ into clusters is obtained by assigning each point to its nearest center. Then, we simply return the solution with the minimum cost. However, note that when $k$ belongs to the range $n/2 \pm o(n)$, $\binom{n}{k} \simeq 2^n$. Thus, the na\"ive algorithm has running time $\Oh^*(2^n)$ in the worst case. 

For many problems, the running time of $\Oh^*(2^n)$ is often achievable by a brute-force enumeration of all the solutions. However, for many \classNP-hard problems, it is often possible to obtain improved running times. The field of exact algorithms for \classNP-hard problems is several decades old. In 2003, Woeginger wrote a survey \cite{woeginger2003exact} on this topic, which revived the field. This eventually led to a plethora of new results and techniques, such as subset convolution \cite{bjorklund2009set}, measure and conquer \cite{fomin2009measure}, and monotone local search \cite{fomin2019monotone}. A detailed survey on this topic can be found in a textbook by Kratsch and Fomin \cite{kratsch2010exact}. We study the aforementioned classical clustering problems from this perspective. In other words, we ask whether the classical clustering problems such as \probMedian and \probMeans admit moderately exponential-time algorithms, i.e., algorithms with running time $c^n \cdot n^{\Oh(1)}$ for a constant $c < 2$ that is as small as possible. We indeed answer this question in the affirmative, leading to the following theorem. 

\begin{theorem}\label{thm:kmed-exp}
		There is an exact algorithm for \probMedian (\probMeans) in time $(1.89)^{n}n^{\Oh(1)}$, where $n$ is the number of points in $X$.
\end{theorem}
To explain the idea behind this result, consider the following fortuitous scenario. Suppose that the optimal solution only contains clusters of size exactly $2$. In this case, it is easy to solve the problem optimally by reducing the problem to finding a minimum-weight matching in the complete graph defining the metric \footnote{Note that the cluster-center always belongs to its own cluster, which implies that a cluster of size $2$ contains one \emph{additional} point. This immediately suggests the connection to minimum-weight matching.}. Note that the problem of finding Minimum-Weight Perfect Matching is known to be polynomial-time solvable by the classical result of Edmonds \cite{edmonds1965}. This idea can also be extended if the optimal solution only contains clusters of size $1$ and $2$, by finding matching in an auxiliary graph. However, the idea does not generalize to clusters of size $3$ and more, since we need to solve a problem that has a flavor similar to the $3$-dimensional matching problem or the ``star partition'' problem, which are known to be \classNP-hard \cite{GareyJ79,ChalopinP14,KirkpatrickH83}. Nevertheless, if the number of points belonging to the clusters of size at least $3$ is \emph{small}, one can ``guess'' these points, and solve the remaining points using  matching. However, the number of points belonging to the clusters of size at least $3$ can be quite large -- it can be as high as $n$. But note that the number of \emph{centers} corresponding to clusters of size at least $3$ can be at most $n/3$. We show that ``guessing'' the subset of centers of such clusters is sufficient (as opposed to guessing \emph{all} the points in such clusters), in the sense that an optimal clustering of the ``residual'' instance can be found---again---by finding a minimum-weight matching in an appropriately constructed auxiliary graph. 

We briefly explain the idea behind the construction of this auxiliary graph. Note that in order to find an optimal clustering in the ``residual'' instance, we need to figure out the following things: (1) the set of points that are involved in clusters of size $1$, i.e., \emph{singleton} clusters, (2) the pairs of points that become clusters of size $2$, and (3) for each center $c_i$ of a cluster of size at least $3$, the set of at least two additional points that are connected to $c_i$. We find the set of points of type (1) by matching them to a set of \emph{dummy} points with zero-weight edges. The pairs of points involved in clusters of size $2$ naturally correspond to a matching, such that the weight of each edge corresponds to the distance between the corresponding pair of points. Finally, to find points of type (3), we make an appropriate number of \emph{copies} of each guessed center $c_i$ that will be matched to the corresponding points. Although the high-level idea behind the construction of the graph is very natural, it is non-trivial to construct the graph such that a minimum-weight perfect matching in the auxiliary graph exactly corresponds to an optimal clustering (assuming we guess the centers correctly). Thus, this construction pushes the boundary of applicability of matching in order to find an optimal clustering. Since the minimum-weight perfect matching problem can be solved in polynomial time, the running time of our algorithm is dominated by guessing the set of centers of clusters of size at least $3$. As mentioned previously, the number of such centers is at most $n/3$, which implies that the number of guesses is at most $\binom{n}{n/3} \le (1.89)^n$, which dominates the running time of our algorithm. We describe this result in Section \ref{sec:matching}. We complement these moderately exponential algorithms by showing that these running times are asymptotically optimal. Formally, assuming the Exponential Time Hypothesis (ETH), as formulated by Impagliazzo and Paturi \cite{Impagliazzo2001}, we show that these problems do not admit an algorithms running in time $2^{o(n)} \cdot n^{\Oh(1)}$. A formal definition of ETH is given in Section \ref{sec:preliminaries}, and we prove the ETH-hardness result in Section \ref{sec:eth}.

We note that our algorithm as well as the hardness result also holds for \probCenter. However, it is folklore that the \emph{exact} versions of \probCenter and \probDominating are equivalent. Thus, using the currently best known algorithm for \probDominating by Iwata \cite{iwata2011faster}, it is possible to obtain an $\Oh^*((1.4689)^n)$ time algorithm for \probCenter. 

We also consider a ``facility location'' or ``supplier'' version, which is a generalization of the clustering problems defined above. In this setting, we are given a set of clients (or points) $X$, and a set of facilities (or centers) $F$. In general the sets $X$ and $F$ may be different, or even disjoint. In these versions, the set of $k$ centers $C$ must be chosen from $F$, i.e., $C \subseteq F$. We formally state the ``supplier'' version of \probMedian, which we call \probMedianFacility \footnote{We note that a slight generalization of this problem has been considered by Jain and Vazirani \cite{jain2001approximation}, who called it ``a common generalization of $k$-median and Facility Location'', and gave a constant approximation in polynomial time.}.

\defproblema{\probMedianFacility}%
{Given a metric space $(X \cup F,d)$, where $X=\{x_1,\ldots,x_n\}$ of $n$ points, called clients, $F$ is a set of $m$ centers, and a positive integer $k$.}%
{Find a pair $(C,P)$, where $C =\{c_1,\ldots,c_k\}\subseteq F$ of size at most $k$ and $P$ is a partition of $X$ into $k$ subsets $\{X_1,\ldots,X_k\}$ (clusters) such that each client in cluster $X_i$ is assigned  to center $c_i$ so as to minimize the $k$-median cost of clustering, defined as follows: 
	$$\cost(C,P)=\sum_{i=1}^{k}\sum_{x \in X_i}^{}d(c_i,x)$$\vspace{-0.5cm}}
It is also possible to define the analogous versions of \probMeans and \probCenter -- the latter has been studied in the approximation literature under the name of $k$-\textsc{supplier}. In this paper, we show that these ``facility location'' versions of \probMedian/\probMeans/\probCenter are computationally harder, as compared to the normal versions, in the following sense. Consider the concrete example of \probMedian and \probMedianFacility. As mentioned earlier, we beat the ``trivial'' bound of $\Oh(2^n)$, by giving a $\Oh((1.89)^n)$ time algorithm for \probMedian. On the other hand, we show that for \probMedianFacility, it is not possible to obtain a $2^{(1-\epsilon)n} \cdot (mn)^{\Oh(1)}$ time algorithm for any fixed $\epsilon > 0$ (note that $m = |F|$ is the number of facilities and $n = |X|$ is the number of clients). For showing this result, we use the \probSeCoCo, which is a complexity theoretic hypothesis proposed by Cygan et al.\ \cite{Cygan2016}. We match this lower bound by designing an algorithm with running time $2^n \cdot (mn)^{\Oh(1)}$ under some mild assumptions. The details are in Section \ref{sec:subset-conv}. While this algorithm is not obvious, it is a relatively straightforward application of the subset convolution technique. This algorithm also works for the supplier versions of \probMeans and \probCenter; however, again there is a much simpler algorithm for $k$-\textsc{supplier} with a similar running time. 

Finally, note that designing an algorithm for the supplier versions with running time $2^{m} \cdot (mn)^{\Oh(1)}$ is trivial by simple enumeration.  It is not known whether the base of the exponent can be improved by showing an algorithm with running time $(2-\epsilon)^m (mn)^{\Oh(1)}$ for some fixed $\epsilon > 0$, or whether this is not possible assuming a similar complexity-theoretic hypothesis, such as \probSeCoCo, or Strong Exponential Time Hypothesis (SETH). We leave this open for a future work.

%% file: preliminaries.tex
\section{Preliminaries} \label{sec:preliminaries}
We denote by $G=(V(G),E(G))$ a graph with vertex set $V(G)$ and edge set $E(G)$.
Cardinality of a set $S$ denoted by 
$|S|$ is the number of elements of the set. We denote an (undirected) edge between vertices $u$ and $v$ as $uv$. We denote by $N(v)=\{u \in V(G) \mid (u,v) \in E(G)\}$ be the \emph{open neighbourhood} (or simply neighbourhood) of $v$, and let 
$N[v]=N(v) \cup \{v\}$ be the \emph{closed neighbourhood} of $v$.

A \emph{matching} $M$ of a graph $G$ is a set of edges such that no two edges have common vertices.
A vertex $ v\in V(G)$ is said to be saturated by $M$ if there is an edge in $M$ incident to $v$, otherwise it is said to be unsaturated.
We also say that $M$ saturates $v$. 
We say that a vertex $u$ is matched to a vertex $v$ in $M$ if there is an edge $e \in M$ such that $e=(u,v)$. 
A perfect matching in a graph $G$ is a matching which saturates every vertex in $G$.
Given a weight function $w\colon E(G)\rightarrow \mathbb{R}_{\geq 0}$, the minimum weight perfect matching problem is to find a perfect matching $M$
(if it exists) of minimum weight $w(M)=\sum_{e\in M}w(e)$.
It is well known to be solvable in polynomial time by the Blossom algorithm of Edmonds \cite{edmonds1965}.


A $q$-CNF formula $\phi = C_1 \wedge \ldots \wedge C_m$ is a boolean formula over $n$ variables $\mathcal{X} = \{x_1, x_2, \ldots, x_n\}$, such that each clause $C_i$ is a disjunction of at most $q$ literals of the form $x_i$ or $\neg x_i$, for some $1 \le i \le n$. In a  $q$-SAT instance we are given a $q$-CNF formula $\phi$, and the question is to decide whether $\phi$ is satisfiable. Impagliazzo and Paturi \cite{Impagliazzo2001} formulated the following hypothesis, called Exponential Time Hypothesis. Note that this ETH is a stronger assumption than $\classP \neq \classNP$.

\emph{Exponential Time Hypothesis (ETH)} states that $q$-SAT, $q \geq 3$ cannot be solved within a running time of $2^{o(n)}$ or $2^{o(m)}$, where  $n$ is the number of variables and $m$ is the number of clauses in the input $q$-CNF formula. 


%% file: matching.tex
\section{Proof of Theorem 1 } \label{sec:matching}
Before delving into the proof of Theorem \ref{thm:kmed-exp}, we discuss the approach at a high level. We begin by ``guessing'' a subset of centers from an (unknown) optimal solution. For each guess, the problem of finding the best (i.e., minimum-cost) clustering that is ``compatible'' with the guess is reduced to finding a minimum weight perfect matching in an auxiliary graph $G$. The graph $G$ is constructed in such a way that this clustering can be extracted by essentially looking at the minimum-weight perfect matching. Note that \probPM problem is well known to be solvable in polynomial time by the Blossom algorithm of Edmonds \cite{edmonds1965}. Finally, we simply return a minimum-cost clustering found over all guesses. 

Let us fix some optimal $k$-median solution and let $k_1^{\ast}$, $k_2^{\ast}$ and $k_3^{\ast}$ be a partition of $k$, where $k_1^{\ast}:$ the number of clusters of size exactly 1, call \emph{Type1}; 
$k_2^{\ast}:$ the number of clusters of size exactly $2$, call \emph{Type2}; and 
$k_3^{\ast}:$ the number of clusters of size at least $3$, call \emph{Type3}.
Let $C_3^{\ast} \subseteq X$ be $\emph{Type3}$ centers, and say $C_3^{\ast} =\{c_1,\ldots,c_{k^{\ast}_3}\}$. Observe that number of clusters with $\emph{Type3}$ centers is at most $\frac{n}{3}$.
Suppose not, then the number of clusters with $\emph{Type3}$ centers is greater than $\frac{n}{3}$.
Each $\emph{Type3}$ cluster contains at least three points.
This contradicts that the number of input points is $n$.

\medskip\noindent\textbf{Algorithm.} First, we guess the partition of $k$ into $k_1$, $k_2$, $k_3$ as well as a subset $C_3 \subseteq X$ of size at most $n/3$. For each such guess $(k_1, k_2, k_3, C_3)$, we construct the auxiliary graph $G$ (as defined subsequently) corresponding to this guess, and compute a minimum weight perfect matching $M$ in $G$. Let $M^*$ be a minimum weight perfect matching over \emph{all} the guesses. We extract the corresponding clustering $(C^*, P^*)$ from $M^*$ (also explained subsequently), and return as an optimal solution of the given instance. 

\medskip\noindent\textbf{Running time.} Note that there are at most $\Oh(k^2)$ tuples $(k_1, k_2, k_3)$ such that $k_1 + k_2 + k_3 \le k$ (note that $k_i$'s are non-negative integers). Furthermore, there are at most $\sum_{i = 0}^{n/3} \binom{n}{i} \le (1.89)^n$ subsets of $X$ of size at most $n/3$. Finally, constructing the auxiliary graph, and finding a minimum-weight perfect matching takes polynomial time. Thus, the running time is dominated by the number of guesses for $C_3$, which implies that we can bound the running time of our algorithm by $\Oh^*((1.89)^n)$.

\medskip\noindent\textbf{Construction of Auxiliary Graph.}
From now on assume that our algorithm made the right guesses, i.e., suppose that $(k_1,k_2,k_3)=(k_1^{\ast}, k_2^{\ast}, k_3^{\ast})$ and $C_3^{\ast}=C'$.
Then, we initialize the $\emph{Type3}$ centers by placing each center from $C'$ into a separate cluster.
At this point, to achieve this, we reduce the problem to the classical \probPM on an auxiliary graph $G$, which we define as follows.
(See Figure~\ref{fig:graph} for an illustration of the construction).
\begin{itemize}
	\item For each $i \in \{1,\ldots,k_3\}$, construct a set of $s=n-k_3-2k_2-k_1$ vertices $C_i=\{c^i_1,\ldots,c^i_s\}$.
	Denote $W=\cup_{i=1}^{k}C_i$; the block of vertices $C_i$ corresponds to center $c_i$.
	\item Let $Y=X\setminus C'$, that is, a set consisting of unclustered points in $X$. 
	Observe $|Y|=n-k_3$.
	Denote $Y=\{y_1,\ldots,y_{(n-k_3)}\}$.
	For simplicity, we slightly abuse the notation by keeping the vertices in $G$ same as points in $Y$. That is,
	for each $i \in \{1,\ldots,(n-k_3)\}$, place a vertex $y_i$ in the set $Y$.
	Make each $y_i$ adjacent to all vertices of $W$.
	\item For each $i \in \{1,\ldots,k_1\}$, construct an auxiliary vertex $u_i$. 
	Denote $\uiso=\{u_1,\ldots,u_{k_1}\}$.
	Make each $u_i$ adjacent to every vertex of $Y$.
	\item Construct a set of $s(k_3-1)$ vertices, $\zfill=\{z_1,\ldots,z_{s(k_3-1)}\}$, that we call fillers and make vertices of $\zfill$ adjacent to the vertices of $W$.
\end{itemize}
We define edge weights. For an edge $(u, v) \in E(G)$, we will use $w(u, v)$ to denote $w((u, v))$ to avoid clutter.
\begin{itemize}
	\item For every $i \in \{1,\ldots,(n-k_3)\}$ and every $j \in \{1, \ldots ,k_3 \}$ set $w(y_i, c^j_h)=d(y_i,c_j)$ for $h \in \{1,\ldots,s\}$, i.e, weight of all edges joining $y_i$ in $Y$ with the vertices of $C_i$ corresponding to center $c_j$.
	\item For every $i,j \in \{1,\ldots,n-k_3\}$, $i \neq j$, set $w(y_i, y_j)=d(y_i,y_j)$, i.e, the weight of edges between vertices of $Y$. 
	\item For every $i \in \{1,\ldots,k_1\}$ and $j \in \{1,\ldots,(n-k_3)\}$, set $w(u_i, y_j)=0$, i.e., the edges incident to the vertices of $\uiso$ have zero weights.
	\item For every $i \in \{1,\ldots,s(k_3-1)\}$ and $j \in \{1, \ldots,k_3\}$,
	$w(z_ic^j_h)=0$, for $h \in \{1,\ldots,s\}$, i.e., the edges incident to the fillers have zero weights.
\end{itemize}

\begin{lemma}
	The graph $G$ has a perfect matching.
\end{lemma}
\begin{proof}
	We construct a set $M \subseteq E(
	G)$ that saturates every vertex in $G$.
	
	Note that $|\uiso|<|Y|$ and every vertex of $\uiso$ is adjacent to every vertex of $Y$. 
	Therefore, we can  construct $M_1 \subseteq E(G)$ by arbitrarily mapping each vertex of $\uiso$ to a distinct vertex of $Y$.
	Clearly, $M_1$ is matching saturating vertices of $\uiso$ .
	Since $|\uiso|=k_1$, $M_1$ saturates $k_1$ vertices of $Y$.
	Denote by $Y'$ the set of vertices of $Y$ that are not saturated by $M_1$.
	Observe $|Y'|=s+2k_2$.
	
	Every vertex of $\zfill$ is adjacent to every vertex of $W$ and $|\zfill|<|W|$.
	Construct $M_{2}\subseteq E(G)$ by arbitrarily mapping each vertex of $\zfill$ to a distinct vertex of $W$.
	Thus, $M_{2}$ is a matching which saturates every vertex of $\zfill$ and since $|\zfill|=s(k_3-1)$, it also saturates $s(k_3-1)$ vertices of $W$.
	Denote by $W'$ the set of vertices of $W$ that is not saturated by $M_2$.
	Observe $|W'|=s$.
	Recall, every vertex of $W'$ is adjacent to every vertex of $Y'$ and note that $|W'|<|Y'|$.
	Therefore, construct $M_3 \subseteq E(G)$ by arbitrarily matching each vertex of $W'$ with a distinct vertex of $Y'$. 
	
	Thus, the matching $M_3$ saturates $s$ vertices in both the sets $W'$ and $Y'$.
	Denote $M'=M_1 \cup M_{2} \cup M_3$.
	Clearly, the vertices of $\uiso$, $W$ and $\zfill$ are saturated by $M'$.
	
	Denote by $Y''=Y \setminus Y'$ the set of vertices of $Y$ that are not saturated by $M'$.
	Note that $|Y''|=2k_2$. 
	Consider $M_4 \subseteq E(G)$ which maps these $2k_2$ vertices to each other.
	We set $M=M' \cup M_4$. 
	It is easy to see that $M$ is a perfect matching.
\end{proof}
\begin{figure}
	\centering
	\includegraphics[scale=0.87]{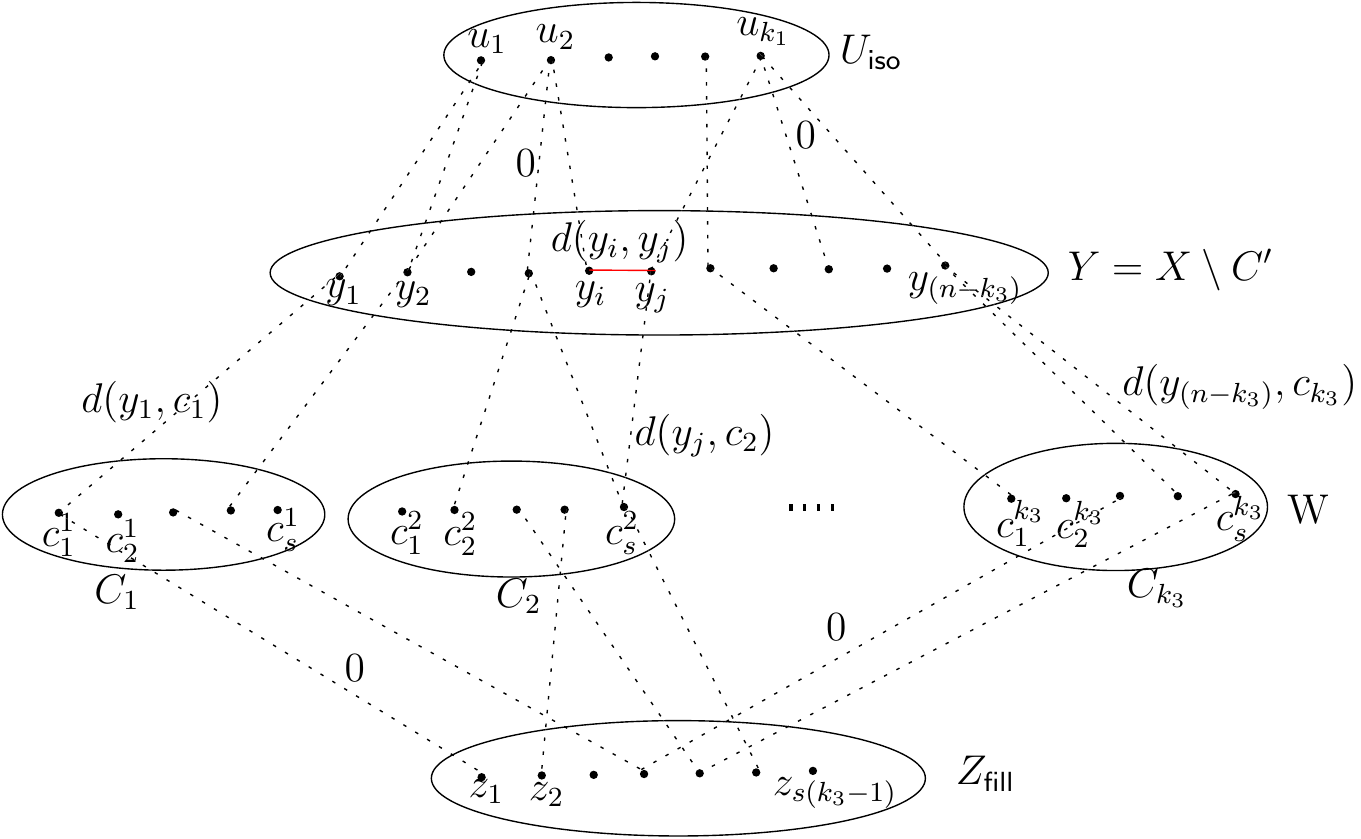}
	\caption{
		Illustration of the graph $G$ produced in the reduction from \probMedian to \probPM. 
		To avoid clutter, we only show some representative edges. Recall that we guess the set of $k_3$ centers of type 3, and corresponding to each such center $c_i$, we add a set $C_i$ consisting of $s$ copies corresponding to that center. Next, we have the set $Y$ corresponding to $n-k_3$ unclustered points. Finally, $\uiso$ and $\zfill$ consist of auxiliary vertices in order to ensure a perfect matching. The weights of vertices among $Y$ correspond to the corresponding original distance; whereas the weight of an edge between $y_\ell \in Y$, and a copy $c_i^j$ corresponding to a type $3$ center $c_i$ is defined to be $d(y_\ell, c_i)$. The weights of all other edges are equal to zero.}
	\label{fig:graph}
\end{figure}
We next show one-to-one correspondence between perfect matchings of
G and $k$-median clusterings of $X$.

\begin{lemma}
	Let $\OPT_{\mm}(G)=$weight of minimum weight perfect matching, and $\OPT_{\kmed}(X)=$ optimal clustering cost of $k$-median clustering of $X$. 
	Then, $\OPT_{\mm}(G)=\OPT_{\kmed}(X)$.
\end{lemma}

\begin{proof}
	\noindent In the forward direction, 
	let $M$ denote a minimum weight perfect matching $M \subseteq E(G)$.
	We construct a $k$-median clustering of $X$ of same cost.
	
	Observe that each vertex of $\zfill$ is only adjacent to the vertices of $W$ and $|\zfill|<|W|$.
	Let $W_1 \subseteq W$ be a set of vertices matched to vertices of $\zfill$.
	Since $G$ has a perfect matching, it saturates $\zfill$, where $
	|\zfill|=s(k_3-1)$.
	Then, $|W_1|=s(k_3-1)$.
	Let $W_2=W\setminus W_1$ be set of vertices matched to vertices of $Y$.
	Clearly,  $|W_2|=s$.

	For every $i \in \{1,\ldots,(n-k_3)\}$, vertex $y_i \in Y$ is saturated by $M$.
	Therefore, we construct the $k$-median clustering $\{X_1,\ldots,X_k\}$ of $X$, where each $X_i \in \{\emph{Type1},\emph{Type2},\emph{Type3}\}$, for $i \in \{1,\ldots,k\}$ as follows.
	
	Let  $Y'\subseteq Y$ be the set of vertices that are matched to vertices of $\uiso$ in $M$, where $|\uiso|=k_1<|Y|$.
	Corresponding to each such vertex in $Y'$, select a center in the solution $C$,  call $C_{\emph{Type1}}=\{c^1_{\emph{Type1}},\ldots,c_{\emph{Type1}}^{k_1}\}$.
	Correspondingly, also construct a singleton cluster $X_i=\{c_{\emph{Type1}}^i\}$, for $i \in \{1,\ldots,k_1\}$.
	Let  $X_{\emph{Type1}}$ denote set of all $\emph{Type1}$ clusters.

	We now construct \emph{Type3} clusters: Let $Y_i'\subseteq Y$ be the set of  vertices matched to set $C_i$, for $i \in \{1,\ldots,k_3\}$ in $M$.
	Consider $X_i=Y_i' \cup \{c_i\}$. 
	Clearly, $X_i$, for $i \in \{1,\ldots,k_3\}$ corresponds to $\emph{Type3}$ clusters in  $X$. 
	Let  $X_{\emph{Type3}}$ denote set of all $\emph{Type3}$ clusters.
	Recall, we already guess set $C'=\{c_1,\ldots,c_{k_3}\}$, that is, $\emph{Type3}$ centers correctly.
	
	Lastly, we construct clusters of $\emph{Type2}$.
	Denote by $Y''$ set of unclustered points in $Y$.
	Observe these points form a set of $k_2$ disjoint edges in $M$.
	Arbitrarily, select one of the endpoint of each edge as a center in the solution $C$, call $C_{\emph{Type2}}=\{c^1_{\emph{Type2}},\ldots,c_{\emph{Type2}}^{k_2}\}$.  
	That is, for an edge $(y_1, y_2) \in M$, where $y_1,y_2 \in Y''$, select center as $y_1$ or $y_2$.
	Then construct a cluster $X_i$, for $i\in \{1,\ldots,k_2\}$ by placing both the endpoints of the edge in the same cluster.
	Denote by $X_{\emph{Type2}}$ the set of all $\emph{Type2}$ clusters.
	
	Clearly,  $X_i \in \{Type1, Type2,Type3\}$, for $i \in \{1,\ldots,k\}$ is a partition of $X$.
	Note, since $\emph{Type1}$ clusters are isolated points, therefore, they contribute zero to the total cost of clustering.
	Now we upper bound the cost of the obtained $k$-median clustering:
	\begin{align*}
		\sum_{i=1}^k\sum_{x\in X_i}d(c_i,x)=
		\sum_{i=1}^{k_2}\sum_{ y \in X_{\emph{Type2}}}^{} d(c^i_{\emph{Type2}},y)+\sum_{i=1}^{k_3}\sum_{ y \in X_{\emph{Type3}}}^{} d(c_i,y) =\OPT_{\mm}(G).
	\end{align*}
	
 For the reverse direction, consider a $k$-median clustering  $\{X_1,\ldots,X_k\}$ 
	of $X$ 	into \linebreak
	 $\{\emph{Type1},\emph{Type2},\emph{Type3}\}$ clusters of $X$ such that
	$|\emph{Type1}|=k_1$, $|\emph{Type2}|=k_2$ and $|\emph{Type3}|=k_3$ and $C'=\{c_1,\ldots,c_{k_3}\}$, that is, centers of $\emph{Type3}$ clusters with $\OPT_{\kmed}(X)$.
	We construct a perfect matching $M \subseteq E(G)$ of $G$ as follows.
	
	Observe that each $\emph{Type1}$ cluster is a singleton cluster.
	Construct $M_1 \subseteq E(G)$ by 
	iterating over each singleton vertex in $Y$ correspond to each cluster and matched it to a distinct vertex in $\uiso$.
	Since $|\emph{Type1}|=|\uiso|=k_1$, $M_1$ is a matching saturating set $\uiso$.
	Also, $M_1$ saturates $k_1$ vertices in $Y$.
	
	Corresponding to each $\emph{Type2}$ cluster, construct $M_2 \subseteq E(G)$ by adding an edge between both the end vertices in $Y$.
	Clearly, $M_2$ is a disjoint set of $k_2$ edges in $G$ and saturates $2k_2$ vertices in $Y$.
	
	Denote $Y'\subseteq Y$ be the set of vertices matched by $M_1 \cup M_2$.
	Clearly, $|Y'|=k_1+2k_2$.
	Let $Y''=Y\setminus Y'$ be the set of remaining unmatched vertices in $Y$.
	Then, $|Y''|=|Y|-|Y'|=n-k_3-2k_2-k_1=s$.
	
	Note, we already guessed $C'=\{c_1,\ldots,c_{k_3}\}$ and we have a cluster $X_i$ corresponding to each $C_i$, for $i \in \{1,\ldots,k_3\}$.
	Construct $M_3 \subseteq E(G)$ by matching each vertex of $X_i \setminus \{c_i\}$ in $Y''$ to a distinct copy of $c_i$ in $W$.
	Since $|Y''|<|W|$, $M_3$ saturates $Y''$.
	Let $W_1 \subseteq W$ be the set of vertices saturated by $M_3$.
	Note that $|Y''|=s$, then $|W_1|=s$.
	Let $W_2= W\setminus W_1$ be the set of vertices not saturated by $M_3$, where $|W|=sk_3$.
	Then, $|W_2|=s(k_3-1)$.
	Every vertex of $\zfill$ is only adjacent to every vertex of $W$ (in particular of $W_2$).
	We construct $M_4 \subseteq E(G)$ by matching each vertex of $\zfill$ to a distinct vertex of $W_2$.
	Since $|\zfill|=|W_2|=s(k_3-1)$, $M_4$ saturates $\zfill$ and $W_2$.
	
	To evaluate the weight of $M$, recall that the edges of $G$ incident to set $\uiso$ and filler vertices $\zfill$ have zero weights, that is, $w(M_1)=w(M_4)=0$.	Then 
	\begin{align*}
		w(M)=w(M_2)+w(M_3)&=\sum_{e \in M_2}^{}w(e)+ \sum_{e \in M_3 }^{}w(e)
		\\&=\sum_{c_i:X_i \in X_{\emph{Type2}}}^{}\sum_{y \in C_i}^{}d(y,c_i)+ \sum_{c_i:X_i \in X_{\emph{Type3}}}^{}\sum_{y \in C_i}^{}d(y,c_i)
		\\&=\OPT_{\kmed}(X).
	\end{align*}
	
	It is straightforward to see that the construction of the graph G from an instance $(X,d)$
	of \probMedian can be done in polynomial time.
	Then, because a perfect matching of minimum weight of the graph $G$ can be found in polynomial time \cite{edmonds1965} and the total number of guesses is at most $(1.89)^nn^{\Oh(1)}$, 
	\probMedian can be solved exactly in $(1.89)^{n}n^{\Oh(1)}$ time. 
	This completes the proof of the theorem.
\end{proof}
\begin{remark}
	Note that even if the distances satisfy the triangle inequality, the sum of \emph{squares} of distances do not. Nevertheless, our algorithm also works for \probMeans, where we want to minimize the sum of squares of distances; or even more generally, if we want to minimize the sum of $z$-th powers of distances, for some fixed $z \ge 1$. In fact, our algorithm works for non-metric distance functions -- it is easy to modify construction of graph $G$ so that it works with asymmetric distance functions, which are quite popular in the context of asymmetric traveling salesman problem \cite{svensson2020constant,traub2020improved}. Finally, we note that it may be possible to improve the running time (i.e., the base of the exponent) using the metric properties of distances, and we leave this open for a future work. However, in the next section, we show the running time of an exact algorithm cannot be substantially improved, i.e., to $\Oh^*(2^{o(n)})$.
\end{remark}

%% file: eth-hard.tex
\section{ETH Hardness} \label{sec:eth}
In this section, we establish result around the (im)possibility
of solving \probMedian problem in subexponential time in the number of points.
For this, we use the result of Lokshtanov et al. \cite{Lokshtanov2011} which states that, assuming ETH, \probDominating problem cannot be solved in time $2^{o(n)}$ time, where $n$ is the number of vertices of graph.

Given an unweighted, undirected graph $G = (V,E)$, a dominating set $S$ is a subset of $V$ such that each $v \in V$ is dominated by $S$, that is, we either have $v \in S$ or there exists an edge $(uv) \in E(G)$ such that $u \in S$.
The decision version of \probDominating is defined as follows.

\defproblema{\probDominating}%
{Given an unweighted, undirected graph $G(V,E)$, positive integer $k$.}%
{Determine whether $G$ has a dominating set of size at most $k$.}

Lokshtanov et. al \cite{Lokshtanov2011} proved the following result.

\begin{proposition}[{\cite{Lokshtanov2011}}]{\label{proposition:EthlowerboundforDom}}
	Assuming ETH, there is no $2^{o(n)}$ time algorithm for \probDominating problem, where $n$ is the number of vertices of $G$ .
\end{proposition}

We use this known fact  about \probDominating to prove the following.

\begin{theorem} \label{thm:k-median-eth}
\probMedian cannot be solved in time $2^{o(n)}$ time  unless the exponential-time hypothesis fails, where $n$ is the number of points in $X$.
\end{theorem}

\begin{proof}
We give a reduction from \probDominating problem to \probMedian problem. Let $(G = (V, E),k)$ be the given instance of \probDominating. We assume that there is no dominating set in $G$ of size at most $k-1$. This assumption is without loss of generality, since we can use the following reduction iteratively for $k' = 1, 2, \ldots, k$, which only incurs a polynomial overhead.

Now we construct an instance $(X, d)$ of \probMedian as follows. First, let $X = V(G)$, i.e., we treat each vertex of the graph as a point in the metric space, and we use the terms vertex and point interchangeably. Recall that the graph $G = (V, E)$ is unweighted, but we suppose that the weight of every edge in $E(G)$ is $1$. Then, we let $d$ be the shortest path metric in $G$. The following observations are immediate.

\begin{observation}\label{obs:domset-metric}\ 
	\begin{itemize}
		\item For all $u \in V(G)$, $d(u, u) = 0$.
		\item For all distinct $u, v \in V(G)$, $d(u, v) = 1 \iff (u,v) \in E(G)$, and $d(u, v) \ge 2 \iff (u, v) \not\in E(G)$.
	\end{itemize}
\end{observation}

We now show that there is a dominating set of size $k$ iff there is a $k$-median clustering of cost exactly $n-k$.

In the forward direction, let $S \subseteq V(G)$ be a dominating set of size $k$. We obtain the corresponding $k$-median clustering as follows. We let $S = \{c_1, c_2, \ldots, c_k\}$ to be the set of centers. For a center $c_i \in S$, we define $X'_i = N[c_i]$. Since $S$ is a dominating set, every vertex in $V(G) \setminus S$ has a neighbor in $S$. Therefore, $\bigcup_{1 \le i \le k} X'_i = V(G)$. Now, we remove all \emph{other} centers except for $c_i$ from the set $X'_i$. Furthermore, if a vertex belongs to multiple $X'_i$'s, we arbitrarily keep it only a single $X'_i$. Let $\{X_1, X_2, \ldots, X_k\}$ be the resulting partition of $V(G)$. Observe that in the resulting clustering, centers pay a cost of zero, whereas every other vertex has a center at distance $1$. Therefore, the cost of the clustering is exactly $n-k$.

In the other direction, let $(S, \{X_1, X_2, \ldots, X_{k}\})$ be a given $k$-median clustering of cost $n-k$. We claim that $S$ is a dominating set of size $k$. Consider any vertex $u \in V(G) \setminus S$, and suppose $u \in X_i$ corresponding to the center $c_i$. Since $u \not\in S$, $d(u, S) \ge d(u, c_i) \ge 1$. This holds for all $n-k$ points of $V(G) \setminus S$. Now, if $u \in C_i$, and $d(u, c_i) > 1$ for some vertex $u \in V(G) \setminus S$, then this contradicts the assumption that the given clustering has cost $n-k$. This implies that every $u \in V(G) \setminus S$ has a center in $S$ at distance exactly $1$, i.e., $u$ has a neighbor in $S$. This concludes the proof.

This reduction takes polynomial time.
Observe that the number of points in the resulting instance is equal to $n$, the number of vertices in $G$.
Therefore, if there is an algorithm for \probMedian with running time subexponential in the number of points $n$ then it would give a $2^{o(n)}$ time algorithm for \probDominating, which would refute ETH, via Proposition~\ref{proposition:EthlowerboundforDom}.
\end{proof}

%% file: secoco-hard.tex
\section{SeCoCo Hardness}
In this section, we consider the variant of \probMedian, which we call \probMedianFacility.
Recall that in this problem, we are given a metric space $(X \cup F,d)$, where $X$ is a set of $n$ clients, $F$ is a set of $m$ centers and integer $k>0$. The goal is to select a set $C \subseteq F$ of $k$ centers and assign each client in $X$ to a center in $C$, such that the $k$-median cost of clustering is minimized.

We show that there is no algorithm solves  \probMedianFacility problem in time $\Oh(2^{(1-\epsilon)n} \poly(m))$, for every fixed $\epsilon>0$.
For this, we use the \probSeCoCo by Cygan et al. \cite{Cygan2016}. 

The decision version of \probSetC problem is defined as follows.

\defproblema{\probSetC}%
{Given a universe $\mathcal{U}=\{u_1,\ldots,u_n\}$ of $n$ elements and a family $\mathcal{S}=\{S_1,\ldots,S_m\}$ of $m$ subsets of $\mathcal{U}$ and an integer $k$  }%
{Determine whether there is a set cover of size at most $k$. 
}

To state \probSeCoCo \cite{Cygan2016} more formally, 
let \probDeltaSetC denote the \probSetC problem where all the sets have size at most $\Delta >0$.

\begin{conjecture}
	\probSeCoCo (SeCoCo)\cite{Cygan2016}. For every fixed $\epsilon > 0$ there is $\Delta(\epsilon) > 0$, such that no algorithm (even randomized) solves \probDeltaSetC in time $\Oh(2^{(1-\epsilon)n} \cdot \poly(m))$.
\end{conjecture}

Using this result, we show the following.

\begin{theorem}
	Assuming \probSeCoCo, for any fixed $\epsilon > 0$, there is no $\Oh(2^{(1-\epsilon)n}\cdot\poly(m))$ time algorithm for \probMedianFacility, where $n$ is the number of clients.
\end{theorem}
\begin{proof}
	We give a reduction from \probSetC to \probMedianFacility problem.
	
	Given an instance  $(\mathcal{U},\mathcal{S})$ of \probSetC problem, where $\mathcal{U}=\{u_1,\ldots,u_n\}$ and $\mathcal{S}=\{S_1,\ldots,S_m\}$, such that $S_i \subseteq \mathcal{U}$,
	we create an instance of \probMedianFacility by building a bipartite graph $G=((X\cup F),E)$ as follows.
	\begin{itemize}
		\item For each element $u_i \in U$, we create a client, say $x_i$, for $i \in \{1,\ldots,n\}$. Denote $X=\{x_1,\ldots,x_n\}$.
		\item For each  set $S_i \in \mathcal{S}$, we create a center, say $c_i$, for $i \in \{1,\ldots,m\}$. Denote $F=\{c_1,\ldots,c_m\}$.
		\item For every $i \in \{1,\ldots,n\}$ and every $j \in \{1,\ldots,m\}$, if $u_i \in S_j$, then connect corresponding $x_i$ and $c_j$  with an edge of weight $1$, i.e., client $x_i$ pays cost $1$ when assigned to facility $c_j$.
	\end{itemize}
	This finishes the construction of $G$.
	Now, let $d$ be the shortest path metric in graph $G$.
	
	We show that there is set cover of size  at most $k$ if and only if there is $k$-median clustering of cost $n$.
	
	In the forward direction, assume there is a set cover $\mathcal{S'}\subseteq \mathcal {S}$ of size at most $k$.
	Assume $\mathcal{S'}=\{S_1,\ldots,S_k\}$.
	For a set $S_i$, we make the corresponding vertex $c_i \in F$ a center. 
	Then, we create its corresponding cluster $X_i$ as follows.
	We add all the points $x_j$ such that $(c_ix_j) \in E$.
	Finally, we make the clusters ${X_i}$ pairwise disjoint, by arbitrarily choosing exactly one cluster for every client, if the client is present in multiple clusters.
	Clearly, $\{X_1,\ldots,X_k\}$ is a partition of $X$.
	We now calculate the cost of the obtained $k$-median clustering.
	$$\sum_{i=1}^{k}\sum_{x \in X_i}d(c_i,x)=\sum_{i=1}^{k}|X_i|= n.$$
	
	In the reverse direction, suppose there is a $k$-median clustering $\{X_1,\ldots,X_k\}$ of $X$ of cost $n$.
	Let $C=\{c_1,\ldots,c_k\} \subseteq F$ be a set of centers.
	Every client must be at distance at least $1$ from its corresponding center.
	We claim that each client in a cluster is at distance exactly $1$ from its corresponding center.
	Suppose not, then there exists a client with distance strictly greater than $1$ from its center. 
	The total number of clients is $n$.
	This contradicts that the cost of $k$-median clustering is $n$.
	Thus, every element is chosen in some set corresponding to set $C$.
	Therefore, a subfamily  $\mathcal{S'}\subseteq \mathcal{S}$ corresponding to set $C$ forms a cover of $\mathcal{U}$.
	Since $|C|=k$, $\mathcal{S'}$ is a cover of $\mathcal{U}$ of size at most $k$.

	Clearly, this reduction takes polynomial time.
	Furthermore, observe that the number of clients in the resulting instance is same as the number of elements in $\mathcal{U}$.
	Therefore, if there is an $\Oh(2^{(1-\epsilon)n}\cdot\poly(m))$ time algorithm for \probMedianFacility then it would give a $\Oh(2^{(1-\epsilon)n}\cdot\poly(m))$ time algorithm for \probSetC, which, in turn, refutes \probSeCoCo.
\end{proof}
We briefly note that the same hardness construction also shows a similar hardness result for the ``supplier'' versions of \probMeans and \probCenter.

%% file: fl-exact.tex
\section{A \texorpdfstring{$2^{n} \cdot \poly(m, n)$}{2\^n poly(m, n)} time Algorithm for \texorpdfstring{$k$}{k}-Median Facility Location} \label{sec:subset-conv}

Let $(X \cup F, d)$ be a given instance of \probMedianFacility, where $n = |X|$ denotes the number of clients, and $m = |F|$ denotes the number of centers. In this section, we give a $2^{n} \cdot\poly(m, n)$-time exact algorithm, under a mild assumption that any distance in the input is a non-negative integer that is bounded by a polynomial in the input size. \footnote{Since the integers are encoded in binary, this implies that the length of the encoding of any distance is $\Oh(\log(m) + \log(n))$.} Let $M \coloneqq n \cdot D$, where $D$ denotes the maximum inter-point distance in the input. Note that $M = \poly(m, n)$.

We define $k$ functions $\cost_1, \cost_2, \ldots, \cost_k: 2^X \to M$, where $\cost_i(Y)$ denotes the minimum cost of clustering the clients of $Y$ into at most $i$ clusters. In other words, $\cost_i(Y)$ is the optimal $i$-\textsc{Median Facility Location} cost, restricted to the instance $(Y \cup F, d)$. First, notice that $\cost_1(Y)$ is simply the minimum cost of clustering all points of $Y$ into a single cluster. This value can be computed in $\Oh(mn)$ time by iterating over all centers in $F$, and selecting the center $c$ that minimizes the cost $\sum_{p \in Y} d(p, c)$. Thus, the values $\cost_1(Y)$ for all subsets $Y \subseteq X$ can be computed in $\Oh(2^n mn)$ time. Next, we have the following observation.
\begin{observation}\label{obs:fl-dp}
	For any $Y \subseteq X$ and for any $1 \le i \le k$, $$\cost_i(Y) = \min_{\substack{A \cup B = Y\\A \cap B = \emptyset}} \cost_{i-1}(A) + \cost_1(B).$$
\end{observation}
Note that since we are interested in clustering of $Y$ into \emph{at most} $i$ clusters, we do not need to ``remember'' the set of facilities realizing $\cost_{i-1}(A)$ and $\cost_1(B)$ in Observation \ref{obs:fl-dp}. Next, we discuss the notion of subset convolution that will be used to compute $\cost_{i}(\cdot)$ values that is faster than the na\"ive computation.

\medskip\noindent\textbf{Subset Convolutions.} Given two functions $f, g: 2^X \to \mathbb{Z}$, the \emph{subset convolution} of $f$ and $g$ is the function $(f * g) : 2^X \to \mathbb{Z}$, defined as follows.
\begin{equation}
	\displaystyle\forall Y \subseteq X: \qquad (f * g)(Y) = \sum_{\substack{A \cup B = Y\\A \cap B = \emptyset}} f(A) \cdot g(B) \label{eqn:subset-conv}
\end{equation}
It is known that, given all the $2^n$ values of $f$ and $g$ in the input, all the $2^n$ values of $f * g$ can be computed in $\Oh(2^n \cdot n^3)$ arithmetic operations, see e.g., Theorem 10.15 in the Parameterized Algorithms book \cite{cygan2015parameterized}. This is known as \emph{fast subset convolution}. Now, let $(f \oplus g)(Y) = \min_{\substack{A \cup B = Y\\A \cap B = \emptyset}} f(A) + g(B)$. We observe that $f \oplus g$ is equal to the subset convolution $f * g$ in the integer min-sum semiring $(\mathbb{Z} \cup \{\infty\}, \min, +)$, i.e., in \Cref{eqn:subset-conv}, we use the mapping $+ \mapsto \min$, and $\cdot \mapsto +$. This, combined with a simple ``embedding trick'' enables one to compute all values of $f \oplus g: 2^X \to \{-N, \ldots, N\}$ in time $2^{n} n^{\Oh(1)} \cdot \Oh(N \log N \log \log N)$ using fast subset convolution -- see Theorem 10.17 of \cite{cygan2015parameterized}. Finally, \Cref{obs:fl-dp} implies that $\cost_i$ is exactly $\cost_{i-1} \oplus \cost_1$, and we observe that the function values are upper bounded by $n \cdot D = M$. We summarize this discussion in the following proposition.
\begin{proposition}\label{lem:subset-conv}
	Given all the $2^n$ values of $\cost_{i-1}$ and $\cost_1$ in the input, all the $2^n$ values of $\cost_i$ can be computed in time $2^n n^{\Oh(1)} \cdot \Oh(M \log M \log\log M)$.
\end{proposition}
Using \Cref{lem:subset-conv}, we can compute all the $2^n$ values of $\cost_2(\cdot)$, using the pre-computed values $\cost_1(\cdot)$. Then, we can use the values of $\cost_2(\cdot)$ and $\cost_1(\cdot)$ to compute the values of $\cost_3(\cdot)$. By iterating in this manner $k-1 \le n$ times, we compute the values of $\cost_k(\cdot)$ for all $2^n$ subsets of $k$, and the overall time is upper bounded by $2^n mn^{\Oh(1)} \cdot \Oh(M \log M \log\log M) $, which is $2^n \cdot \poly(m, n)$, if $M = \poly(m, n)$. Note that $\cost_k(X)$ corresponds to the optimal cost of \probMedianFacility. Finally, the computed values of the functions $\cost_i(\cdot)$ can be used to also compute a clustering $\{X_1, X_2, \ldots, X_k\}$ of $X$, and the corresponding centers $\{c_1, c_2, \ldots, c_k\}$. We omit the straightforward details.
\begin{theorem}\label{thm:}
	\probMedianFacility can be solved optimally in $2^n \cdot \poly(m, n)$ time, assuming the distances are integers that are bounded by polynomial in the input size.
\end{theorem}
Note that the algorithm does not require the underlying distance function to satisfy the triangle inequality. In particular, we obtain an analogous result the ``facility location'' version of the \probMeans objective. Finally, the algorithm works for $k$-\textsc{supplier}, which is a similar version of \probCenter. However, in this case there is a much simpler reduction to \probSetC which gives an $2^n \cdot \poly(m, n)$ time algorithm. For this, we first ``guess'' the optimal radius $r$, and define a set system that consists of balls of radius $r$ around the given centers. We omit the details.
